\documentclass[10pt, twoside,twocolumn]{IEEEtran}

\usepackage{hyperref,amsmath,amsthm,amssymb,cite}
\usepackage{graphicx,color,lipsum,dsfont,tikz,pgfplots}
\usepackage[colorinlistoftodos,draft]{todonotes} 
\usepackage{algorithm}

{}
\newtheorem{lemma}{\bf \noindent Lemma}{}

\IEEEoverridecommandlockouts

\title{Optimal Energy Beamforming under Per-Antenna Power Constraint }

\author{
	\IEEEauthorblockN{Zahra~Rezaei,
		Ehsan~Yazdian,  Foroogh~S.~Tabataba,  and Saeed~Gazor}\thanks{Z.~Rezaei, E.~Yazdian and F.~S.~Tabataba are with the Department of Electrical and Computer Engineering, Isfahan University of Technology, Isfahan 84156-83111, Iran. Email: zahra.rezaei@ec.iut.ac.ir, yazdian, fstabataba@cc.iut.ac.ir}
\thanks{Saeed Gazor is with the ECE Department, Queens' University, Kingston, Canada. Email: gazor@queensu.ca}	
}

\begin{document}

\maketitle

\begin{abstract}
Energy beamforming (EB) is a key technique to enhance the efficiency of wireless power transfer (WPT). In this paper, we study the optimal EB under per-antenna power constraint (PAC) which is more practical than the conventional sum-power constraint (SPC). We consider  a multi antenna energy transmitter (ET) with PAC that broadcasts wireless energy to multiple randomly placed energy receivers (ER)s within its cell area. We consider sum energy maximization problem with PAC and provide the optimal solution structure for the general case. This optimal structure implies that sending one energy beam is optimal under PAC which means that the rank of transmit covariance matrix is one similar to SPC. We also derive closed-form solutions for two special cases and propose two sub-optimal solutions for general case, which performs very close to optimal beamforming.
\end{abstract}


\section{Introduction}\label{sec1}
\IEEEPARstart{R}{ecent} advances in microwave wireless power transfer (WPT) technology enables us to build wireless
powered communication networks (WPCNs), where wireless devices such as smart phones, RF identification (RFID) tags, wearable electronic devices and etc., are powered over the air by wireless power transmitters \cite{survay,oppo}. In WPCNs, the wireless devices harvest the energy from received signal.
 WPCN is more user-friendly and cost-effective by reduceing the need for manual battery replacement/recharging and connection cables. In addition, WPCNs 
 enable on-demand energy delivery and reduce the chance of interruption during operation.
For this reason, RF-enabled WPT has attracted a lot of attention in wireless researches, due to its controllability and reliability (see, e.g., \cite{coll, onebit, multi, retro, disone, feng, chu,enhancing, ramezani}).

The main drawback of WPT is the decay in electromagnetic wave as the transmission distance increases which yields low received energy. Moreover the RF energy is attenuated due to channel fading caused by reflection, scattering, and refraction in propagation environment. The received signal may be very weak, making  difficult to harvest energy from it or to detect. The problem is more challenging  using WPT, since a more significant  signal strength level is required by an energy receiver (ER). For
instance, a typical information receiver can operate even at $-60$ dBm received signal power,
whereas an ER needs up to $-10$ dBm signal power \cite{mimo}. Hence, it is necessary to design more
efficient WPT mechanisms  \cite{enhancing}. 

To efficiently solve this problem, multi-antenna techniques, which have been successfully employed in wireless communication systems to improve the information transmission rate and reliability over wireless channels, have been also proposed for WPT \cite{onebit}. Employing multi-antenna systems can increase the efficiency of the power transfer gain without increasing transmit power and bandwidth. Specifically, deploying multiple antennas at the energy transmitter (ET) enables us to use advanced energy beamforming (EB) techniques to focus the transmitted power toward the desired ERs and thereby to maximize
the received signal amplitude. By optimizing the transmit waveforms, ET could control the collective behaviour of the radiated waveforms causing them to combine coherently at a desired ER \cite{oppo}. The analysis of directional WPT under different scenarios has been studied in \cite{mimo, onebit, chu, coll, multi, disone, retro}.
For the point-to-point MIMO WPT system, it has been shown in \cite{mimo} that the EB is the optimal
solution to maximize the harvested energy by transmitting an energy beam at the ET. The weighted sum-energy maximization problem is formulated in \cite{onebit} for a multiuser MIMO WPT system, which results in an optimal solution similar to \cite{mimo}. 

According to the current hardware technology in multi-antenna systems, a more practical constraint is to have per-antenna power constraint (PAC). This is because that each individual antenna has its own RF power amplifier which has individual power limit \cite{massive, DAS}. Another appealing scenario for the PAC is in the distributed MIMO systems, where  the transmitted antennas are at different physical locations attempting to cooperate in the design of their transmit signals. In these cases, each node has its own  power constraint instead of sharing a total  power budget among different nodes \cite{miso, emil, zhang2017}. We point out that the energy signals at distributed  antennas may be designed offline and stored for real-time transmission, which is in contrast to information signals which are independents stochastic sequences on different distributed ET antennas \cite{coll,zhang2017}.
To the best of our knowledge, the existing literature on EB (e.g., \cite{onebit,multi,mimo}) only assume a  sum-power constraint (SPC). The authors in \cite{coll} have considered a distributed MIMO system for collaborative WPT and derived a closed-form solution for maximizing the sum of powers subject to PAC with only two transmit and receive antennas. They have  shown that the optimal solution is one single energy beam in this situation.

In this paper, we consider sum power maximization problem under PAC and provide the optimal solution structure. We also derive closed-form solutions for two special cases 1) two transmit antennas and arbitrary number of ERs with single antenna and 2) arbitrary number of antennas at ET and one ER. In addition, we propose two sub-optimal solutions for general case. Simulation results show that these sub-optimal solutions are matched closely to optimal numerical results. We show that in the case of PAC, similar to SPC, transmitting only one single energy beam at the ET is optimal for maximizing the sum power in all ERs.
From a practical viewpoint, it is desirable specially in distributed systems, since it means that in each time interval, only one  energy signal should be stored at each distributed ET antenna \cite{coll}. Furthermore, since the optimal covariance matrix is rank-one, our approach can be employed in other applications such as beamforming to maximize the sum rate in a common data multicast system\cite{multicast,multicast2}.

The rest of this paper is organized as follows. Section II presents the WPTN model. Section III studies the sum-energy maximization problem with PAC. Analytical solution of this problem is also provided.  Numerical results are provided in Section IV and the conclusions are drawn in Section \ref{sec:con}.

\section{System Model}\label{sec:model}

\begin{figure}
\centering

\begin{tikzpicture}[even odd rule,rounded corners=2pt,x=12pt,y=12pt,scale=.9]
\draw[thick,fill=red!10] (2,0) rectangle (4,10) node[midway,rotate=90]{Energy Transmitter};
\draw[->,thick] (4,1)--++(1.5,0)--+(0,1);
\draw[thick,fill=gray!15] (5.5,2)--++(.75,0.5)--++(-1.5,0)--++(.75,-.5)--++(0,-.1);
\draw (5.5,4.5) node  {\huge$\vdots$};
\draw[->,thick] (4,6)--++(1.5,0)--+(0,1);
\draw[thick,fill=gray!15] (5.5,7)--++(.75,0.5)--++(-1.5,0)--++(.75,-.5)--++(0,-.1);

\draw[->,thick] (4,9)--++(1.5,0)--+(0,1);
\draw[thick,fill=gray!15] (5.5,10)--++(.75,0.5)--++(-1.5,0)--++(.75,-.5)--++(0,-.1);

\draw [-,fill=blue!30,blue!20] (6.5,4.5) to [bend left=20] ++(3,.5) to [bend left=40]
++(0,-1) to [bend left=20] ++(-3,.5);
\draw [-,fill=blue!30,blue!20] (6.5,6) to [bend left=20] ++(4.5,1.5) to [bend left=40]
++(0,-1) to [bend left=20] ++(-4.5,-.5);

\draw [-,fill=blue!30,blue!20] (7,2.5) to [bend left=20] ++(2,0) to [bend left=40]
++(0,-1) to [bend left=20] ++(-2,1);

\draw[thick,fill=green!10] (10,8) rectangle ++(2,1.5) node[midway]{\footnotesize $\mathrm{ER}_1$}
 ++(0,0) node[right]{\small \shortstack[l]{ Energy \\ Receiver 1}};
\draw[<-,thick,fill=gray!15] (11,9.5)--++(0,.5)--++(.75,0.5)--++(-1.5,0)--++(.75,-.5)--++(0,-.1);

\draw[thick,fill=green!10] (14,5) rectangle ++(2,1.5) node[midway]{\footnotesize $\mathrm{ER}_2$};
\draw[<-,thick,fill=gray!15] (15,6.5)--++(0,.5)--++(.75,0.5)--++(-1.5,0)--++(.75,-.5)--++(0,-.1);

\draw[thick,fill=green!10] (1
6,0) rectangle ++(2,1.5) node[midway]{\footnotesize $\mathrm{ER}_K$};
\draw[<-,thick,fill=gray!15] (17,1.5)--++(0,.5)--++(.75,0.5)--++(-1.5,0)--++(.75,-.5)--++(0,-.1);

\draw[thick,fill=green!10] (10,1) rectangle ++(2,1.5) node[midway]{\footnotesize $\mathrm{ER}_3$};
\draw[<-,thick,fill=gray!15] (11,2.5)--++(0,.5)--++(.75,0.5)--++(-1.5,0)--++(.75,-.5)--++(0,-.1);

\draw (14,2) node {\huge$\ddots$};
\end{tikzpicture}

\caption{Graphical illustration of a multiuser wireless power transfer (WPT) system using a MISO broadcast system.}
\label{fig0}
\centering
\end{figure}
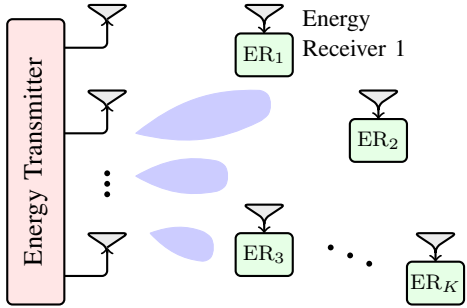
We consider a multiuser MISO broadcast system for WPT as shown in Fig. \ref{fig0}, where one ET with $ N \geq1 $ transmit antennas transfers wireless energy to $K$ single-antenna ERs. 
We assume that the channel response vector from the ET to the $k^{th}$ ER $\mathbf{h}_{k} \in \mathds{C}^{1 \times N}$ has a circularly symmetrical  complex Gaussian distribution. Let $H = [\mathbf{h}_1, \mathbf{h}_2, \cdots, \mathbf{h}_{K}]^T  \in \mathds{C} ^{K \times N} $ denotes the concatenation of all channels, where the $k$-th row is the channel coefficients of the $k^{th}$ ER to ET. Channels are assumed to be quasi-static flat fading, where the channel coefficients of the ET to ERs remain constant within each transmission block and may change from one block to another. In addition, the channel model includes both small scale fading and distance-dependent pathloss components. Duration of each transmission block is $T$ symbols, which is assumed to be sufficiently long for typical low-mobility WPT applications. In addition, it is assumed that perfect channel state information (CSI) is available at ET for designing linear EB.

The received signal at $k^{th}$ ER is
\begin{equation}
\mathbf{y}_k =\mathbf{h}_k \mathbf{x} + n_k , \hspace{0.25 cm} k = 1, \cdots , K,
\end{equation}
where $\mathbf{x}$ is the ET signal and $n_k$ is the additive noise at $k^{th}$ ER. We assume that the harvested power denoted by $E_k$ at the $k^{th}$ ER is fraction of the received RF power.
 Thus, the harvested energy at the $k^{th}$ ER can be expressed as
\begin{equation}
E_k = \rho_k T \mathbb{E} \lbrace | \mathbf{y_k} |^2 \rbrace \approx \rho_k T \mathbb{E} \lbrace | \mathbf{h_k} \mathbf{x} |^2 \rbrace,
\label{2}
\end{equation}
where the constants $\rho_k \in [0 , 1], k = 1 \cdots , K$ represent the energy harvesting efficiency of $k^{th}$ ER. Note that the ER does not need to convert the received RF signal  to the baseband in order to harvest the energy. For convenience, we assume that $\rho_1 = \cdots = \rho_{K} = \rho$ in  this paper. Since $T$ and $\rho$ have fixed values, they do not affect our results, thus, we assume $\rho=1, T=1$sec in the rest of this paper, thus the average energy and power are identical for $T=1$sec. In addition in (\ref{2}), we have ignored the background noise power $\sigma_k^2=\mathbb{E}\{|n_k|^2\}$ since $\sigma_k^2$ is practically insignificant compared to the average received signal power from the viewpoint of WPT \cite{mimo}.
In this case, by defining $Q = \mathbb{E} \lbrace \mathbf{x} \mathbf{x^H} \rbrace$ as the transmit covariance matrix of the energy signals from ET, we can write $E_k\propto  \mathbf{h_k} Q \mathbf{h_k}^H$. Thus the sum of energies harvested  by all ERs are proportional to $\sum_{k=1}^{K} \mathbf{h_k} Q \mathbf{h_k}^H=\mbox{tr}(HQH^H)$ where $\mbox{tr}(.)$ denotes the trace of a matrix.

\section{WPT design under Per-Antenna Power Constraint}\label{sec:prob_sum}
Under the SPC, the total transmit power from ET is limited by $p_t$ which can be allocated arbitrary among the transmit antennas.  This constraint is shown in the matrix form as $\mbox{tr} (Q) \leq p_t$.
In this section, the design objective for $Q$ is to maximize the sum-energy received by all ERs under the per-antenna power constraint.
The diagonal values of $Q$ represents the power transmitted by each antenna which should be bounded by $p_{i}$, i.e., $q_{i,i} \leq p_{i}$ which is clearly a more stringent constraint than $\mbox{tr} (Q) \leq p_t$ \cite{miso}.
 Thus,  our design problem can be formulated as
\begin{align}
\begin{array}{rrclcl}
  (\textrm{P1}): \max_{Q} & \multicolumn{3}{l}{ \mbox{tr} (H Q H^H)} \\
\textrm{s.t.} & q_{i,i} \leq p_i & \forall i \in N \\
&\multicolumn{3}{l}{ Q \succeq 0,\ \ \ Q=Q^H. }
\end{array}
\label{pr4-2}
\end{align}

 Let us express the positive semi-definite matrix $Q$ by its eigenvalue decomposition as $Q = G \Gamma G^H=WW^H$, where $G G^H = I$, $G \in \mathds{C}^{N \times r}$ is the precoding matrix, $\Gamma = \mbox{diag} (\gamma_1, \cdots,\gamma_r)$, with $\gamma_1 \geq \gamma_2 \geq ,\cdots, \geq \gamma_r \geq 0$ are the positive eigenvalues and $W=G\Gamma^{\frac 12}$ \cite{coll}. Using the optimal beamforming matrix, the ET can generate and transmit the  vector
$\mathbf{x} = WS=\sum_{m=1}^{r} \mathbf{w}_m s_m \in \mathds{C}^{N \times 1},$
 where the beam $\mathbf{w}_m \in \mathds{C}^{N \times 1}$    is $m$th column of $W$. The elements of $S=[s_1,\cdots,s_r]$ shall be uncorrelated and should have unit variance $E\lbrace|s_m|^2 \rbrace = 1, \forall m$. These elements  can be generated as independent sequences of arbitrary distribution with zero mean and unit variance \cite{onebit}.
In \cite{onebit,mimo}, it is shown that ignoring the power constraints on each antenna (PAC), the maximum achievable sum energies harvested by  all ERs is $\sup E_t= p_t \xi_1$ which is achieved by $Q=p_t\mathbf{v}_{1} \mathbf{v}_{1}^H$ where $\xi_1$ and $\mathbf{v}_1$ are the largest eigenvalue of $H^HH$ and its corresponding eigenvector, respectively.

We can rewrite the set of constraints $q_{i,i} \leq p_i$ as $\mathbf{e}_i^T Q \mathbf{e}_i\leq p_i$, since $q_{i,i}=\mathbf{e}_i^T Q \mathbf{e}_i$ where $e_i = [0 \cdots 1 \cdots 0]^T$ with 1  as its $i^{th}$ element and $0$ elsewhere. Thus, these PACs and the cost function are all linear in $Q$ \cite{miso}. However the convex constraint $Q\succeq 0$ makes the problem challenging. Therefore, since the problem (P1) is convex, we can employ the existing semi-definite programming (SDP) to solve it. However, the existing iterative methods are either computationally expensive and involve some additive error because of the large number of unknowns. For the special case of $N = K = 2$,  a closed-form analytic solution is found for this problem  in \cite{coll}. However, to the best of our knowledge there is no closed-form solution available for the general case which is the scope of this paper.

\section{Peroperties Of The Optimal Solution}\label{sec:optimal}
\begin{lemma}
The $i$th diagonal value of $Q$ for the solution of problem (P1) must be equal to $q_{i,i} = p_i, i=1, ..., N$, otherwise we can increase the objective function by only increasing $q_{i,i}$ to $ p_i$ while all constraints remain satisfied  \cite{miso}.
\end{lemma}

Employing the above lemma, we only need to find the off-diagonal entries of $Q$. The main complexity here is due to the positive semi-definite constraint (i.e., $Q \succeq 0$).
In this section, we consider the properties of positive semi-definite matrices and find closed-form solutions for two special cases. For the general case, we also propose two heuristic close to optimal solutions which are very close to the results from numerical methods.

A matrix $Q$ is positive semi-definite, if and only if all its eigenvalues are non-negative, which means that the smallest eigenvalue of $Q$ shall be non-negative.
An alternative way to verify $Q$ is positive semi-definite is that its principal minors are all positive semi-definite (a principal minor is obtained by removing a subset of its columns and the same subset of corresponding its rows) \cite{miso}. In the following lemma, we first only use $2 \times 2$ principal minors  of $Q$ and find an important property of the optimal solution.
 We then show that the rank of the optimal solution of problem P1 is one, when relaxing the constraints on higher order principal minors. Interestingly, it  turns out that the solution for the relaxed problem   satisfies $Q\succeq 0$, and thus is also the optimal solution for P1.

\begin{lemma}\label{lemma: 1} {For a given $H \in \mathds{C} ^{K \times N} $ with $\alpha_{i,j} = [H^HH]_{i,j} \neq 0$ for all $i,j \in \{1,\cdots, N\}$},
the amplitude of the off-diagonal entries of $Q$ are $| q_{i,j} | = \sqrt{p_i p_j}$ for all $ i,j \in \{1,\cdots,N\}$ for the solution of the relaxed version of  (P1) where all $1\times 1$ and $2 \times 2$ principal minors  of $Q$ are positive semi-definite.\end{lemma}
\begin{proof}
The $2 \times 2$ principal minor of $Q$ which contains  $i$ and $j$ columns and rows of $Q$  must be positive semi-definite, i.e.,
\begin{equation*}
M_{i,j}(Q) =
\begin{bmatrix}
p_i & & q_{i,j}^*\\
q_{i,j} & & p_j
\end{bmatrix} \succeq 0.
\end{equation*}
This constraint is satisfied if determinant of $M_{(i,j)}$ is non-negative, i.e., $| q_{i,j} | ^2 \leq p_i p_j $.
Moreover, the objective function in \eqref{pr4-2} is the inner product of $Q$ and $H^H H$, i.e.,
$\mbox{tr} ( H Q H^H )= \sum_{\forall i,j} q_{i,j} \alpha_{i,j}$ and only contains $q_{i,j}$ in two of its terms as $q_{i,j}^*\alpha_{i,j}^* + q_{i,j} \alpha_{i,j}=2|q_{i,j}|\ |\alpha_{i,j}| \cos(\sphericalangle q_{i,j}-\sphericalangle \alpha_{i,j})$. 
These two terms are an increasing function of $| q_{i,j} |$ where $(\sphericalangle q_{i,j}-\sphericalangle \alpha_{i,j})$ is an acute angle. Therefore, under the constraint $| q_{i,j} | \leq \sqrt{p_i p_j}$,  these two terms are maximum for $| q_{i,j} | = \sqrt{p_i p_j}$ since for some  $\sphericalangle q_{i,j}$  the angle  $(\sphericalangle q_{i,j}-\sphericalangle \alpha_{i,j})$ is acute.
\end{proof}

Now we will show that the optimal solution for the relaxed problem  using $| q_{i,j} |  = \sqrt{p_i p_j} $ results in positive semi-definite covariance matrix which implies that the resulting solution is an optimal solution for the original problem.
From Lemma~\ref{lemma: 1}, we know the magnitudes $|q_{i,j}|=\sqrt{p_i p_j}$ and only need to find the phase of $q_{i,j}$. A necessary and sufficient condition for $Q\succeq 0$ is that the smallest eigenvalue of $Q$ be non-negative. To use this property, we  associate a matrix dual variable $B$ to the constraint $Q \succeq 0$, and form the Lagrangian as follows (see \cite{dual} for more details)
\begin{equation}
\mathcal{L} (Q , B , C) = - \mbox{tr} (H Q H^H) + tr( C (Q - P)) - \mbox{tr} (B Q),
\label{lag}
\end{equation}
where $P=\mathrm{diag}( p_1, p_2, \cdots, p_{N} )$ denotes the power budgets and $C=\mathrm{diag}( c_1, c_2, \cdots, c_{N})$ is a diagonal Lagrangian multipliers matrix for the PACs.
 By setting the first order derivative of  (\ref{lag}) with respect to $Q$ equal to zero, we obtain
$B = C - H^H H.$
Thus, the off-diagonal entries of $B$ are $b_{i,j}= - \alpha_{i,j}, i \neq j$ where $b_{i,j} = [B]_{i,j}$ for all $i,j \in \{1,\cdots, N\}$ and we should only find the diagonal entries $b_{i,i} = c_{i} - \alpha_{i,i}$. Note that since $B$ is positive semi-definite, $b_{i,i}$s are real and non-negative, thus $c_{i} \geq \alpha_{i,i}$ which yields $C$ be an strictly positive matrix. The complementary slackness KKT conditions yields
\begin{equation}
BQ = QB = 0, \label{BQ}
\end{equation}
which means that $B$ is a hermitian positive semi-definite matrix in the null space of $Q$. By multiplying each row of $B$ in each column of $Q$ and equating it to zero, we can obtain $N \times N$ equations with complex coefficients and variables. The diagonal enties of $Q$ are  $q_{i,i} = p_i\geq 0$, however, the off-diagonal entries $q_{i,j}, i\neq j$ are unknown and complex in general. In contrast, the off-diagonal values of $B$ are known and complex, while its diagonal entries are unknown. Furthermore, since $B$ is hermitian, $b_{i,i}$ is real, and the number of unknown variables and known values are equal. In the following lemma an important property of the solution of the relaxed problem is presented which facilitates approaching the final solution.


\begin{lemma}\label{lemma: 2}
The rank of the solution of for the relaxed problem is one and  can be decomposed as
\begin{equation}
\label{eqn: 8}
Q=\mathbf{w}_1 \mathbf{w}_1^H,
\end{equation}  where $\mathbf{w}_1=[\sqrt{p_1},\sqrt{p_2}e^{j\theta_2},\cdots,\sqrt{p_1}e^{j\theta_N}]^T$ is its dominant eigenvector and $\gamma_1 = \sum_{i=1}^{N} p_i $ is its non-zero eigenvalue.
\end{lemma}
\begin{proof} Using Lemma~\ref{lemma: 1} and by multiplying the $i$th row of $B$ in $N-1$ columns of $Q$ and equating them to zero as in \eqref{BQ} for $i=1,2, \cdots, N$, we can obtain
\begin{align}
b_{i,i} &= c_{i} - \alpha_{i,i}= \sum_{j=1 , i \neq j}^{N} \frac{\sqrt{p_j}}{\sqrt{p_i}} | \alpha_{i,j} | e^{ j (\sphericalangle \alpha_{i,j} - \sphericalangle q_{i,j})}
\label{ppp1}
\\
\sphericalangle q_{i,j} &= \sphericalangle q_{1,j} - \sphericalangle q_{1,i} \quad \forall i \neq j.
\label{ppp3}\end{align}
Substituting $| q_{i,j} | = \sqrt{p_i p_j}$ and (\ref{ppp3}), we can write  $Q$ as
\begin{equation*}
\begin{bmatrix}
p_{1}  \hspace*{-5mm}&\hspace*{-5mm}  \sqrt{p_1p_2} e^{ j \sphericalangle q_{1,2} } \hspace*{-5mm}& \hspace*{-5mm}  \ldots &\hspace*{-5mm}  \sqrt{p_1 p_N} e^{ j \sphericalangle q_{1,N} } \\
\sqrt{p_1p_2}e^{- j \sphericalangle q_{1,2} }  \hspace*{-5mm}&\hspace*{-5mm}  p_{2}  \hspace*{-5mm}&\hspace*{-17mm}  \ldots
 \hspace*{-2mm}&\hspace*{-8mm}  \sqrt{p_2 p_N}e^{ j (\sphericalangle q_{1,N} - \sphericalangle q_{1,2})} \\
\vdots  \hspace*{-5mm}&\hspace*{-5mm}  \vdots  \hspace*{-5mm}&\hspace*{-5mm}  \ddots
 \hspace*{-5mm}&\hspace*{-5mm}  \vdots \\
\sqrt{p_1 p_N} e^{- j \sphericalangle q_{1,N} }  \hspace*{-2mm}&\hspace*{-2mm}  \sqrt{p_2 p_N}e^{ j (\sphericalangle q_{1,2} - \sphericalangle q_{1,N})}  \hspace*{-5mm}&\hspace*{-1mm}  \ldots
 \hspace*{-7mm}&\hspace*{-5mm}  p_{N}
\end{bmatrix}.
\end{equation*}
The above matrix  is equal to $Q= \mathbf{w}_1 \mathbf{w}_1^H$
where $\mathbf{w}_1 =
[\sqrt{p_1}e^{j\sphericalangle q_{1,1}},\sqrt{p_2}e^{j\sphericalangle q_{2,1}},\cdots,\sqrt{p_1}e^{j\sphericalangle q_{1,N}}]^T$. Thus, the optimal input covariance matrix is rank-one. The proof is complete using the eigenvalue decomposition theorem. The non-zero eigenvalue of $Q$
as $\gamma_1 =\|\mathbf{w}_1\|^2= \sum_{i=1}^{N} p_i > 0$ and so optimal $Q$ is positive semi-definite.
\end{proof}
Since the optimal solution in \eqref{eqn: 8} for the relaxed problem is positive semi-definite, it is also an optimal solution for (P1).
Lemma~\ref{lemma: 2} implies that the optimal EB invests all its transmit power only in one beamformer, where $\sqrt{p_i} e^{j \theta_i}s_1$ is the transmitted signal from $i^{th}$ antenna, thus we only need to optimize the unknown phases $\theta_i$ for $i=2,\cdots,N$. Note that $\theta_1$ can be considered as absorbed or accounted in $s_1$.

Interestingly, for the special case of $p_1=p_2= \cdots =p_{N}= \frac{p_t}{N}$, the optimal solution is an equal energy transmission scheme. Since $b_{i,i}$s are real,  the imaginary part of $b_{i,i}$s are zero which yields  a set of equations for optimal $\theta_i$ as in \eqref{ppppp1}.
\begin{figure*}\begin{eqnarray}\label{ppppp1}\begin{array}{c}
{\sqrt{p_2}| \alpha_{1,2}|}  \sin(\theta_2 - \theta_1 + \sphericalangle \alpha_{1,2}) + \cdots +  {\sqrt{p_N}| \alpha_{1,N}| } \sin(\theta_{N} - \theta_1 + \sphericalangle \alpha_{1,N}) = 0,

\\
 {\sqrt{p_1}} | \alpha_{1,2}| \sin(\theta_1 - \theta_2 - \sphericalangle \alpha_{1,2}) + \cdots +  {\sqrt{p_N}} | \alpha_{2,N}| \sin(\theta_2 - \theta_{N} - \sphericalangle \alpha_{2,N}) = 0,
\\
\vdots\label{ppppN}
\\
 {\sqrt{p_1}| \alpha_{1,N}|}  \sin(\theta_1 - \theta_N - \sphericalangle \alpha_{1,N}) + \cdots +  {\sqrt{p_{N-1}}| \alpha_{N-1,N}|}  \sin(\theta_{N-1} - \theta_{N} - \sphericalangle \alpha_{N-1,N}) =0.\end{array}
\end{eqnarray}\end{figure*}
We can alternatively obtain \eqref{ppppp1} by substituting \eqref{eqn: 8} in (P1) and set the derivatives with respect to $\theta_i$ to zero.
The $N$ equations in \eqref{ppppp1} involve $N$ unknown variables  $\theta_1, \theta_2, \cdots, \theta_{N}$. However, since only $N-1$ of them are linearly independent, without loss of optimally we set $\theta_{1} =0$ and so the number of independent equations become equal to the number of unknowns.
It is not easy to derive closed-form exact expressions for $\theta_i$ for arbitrary $N$. However, optimal $\theta_2, \theta_3, \cdots, \theta_{N}$ can be easily calculated by using standard numerical methods such as Newton. Moreover, we find  closed-form  solution for some special cases.

\subsubsection{Case of two transmit antennas}
In this case according to \eqref{ppppp1} from $ \frac{\sqrt{p_2}}{\sqrt{p_1}}| \alpha_{1,2}| \sin(\theta_1 - \theta_2 - \sphericalangle \alpha_{1,2})=0$, we obtain $\theta_1-\theta_2=\sphericalangle \alpha_{1,2}$.

\subsubsection{Case of one ER and multiple transmit antennas} For $K = 1$, the rank of $H^HH$  is one. Thus $\sphericalangle\alpha_{i,j} = \sphericalangle h_j - \sphericalangle h_i$. Therfore, we can choose the trivial optimal value for $\theta_i $ as $ -\sphericalangle h_i$. This is because letting $\theta_i = -\sphericalangle h_i$ yields $\sin(0) =0$ that satisfies all $N-1$ equations in \eqref{ppppp1}. This solution is identical to the one presented in \cite{miso} . The author in \cite{miso} develops a closed-form solution for capacity of MISO channel with PAC. In fact, the problem (P1) is a generalization of problem
(6) in \cite{miso} which only considers the case of $K =1$.

\subsubsection{Case of 3 transmit antennas $N=3$ and arbitrary $K$}
For $N = 3$ antennas, we can write
\begin{equation}
Q= \!\begin{bmatrix}
p_{1}  \hspace*{-1mm}&\hspace*{-1mm}  \sqrt{p_1p_2} e^{ j \theta_2 }
 \hspace*{-1mm}&\hspace*{-1mm}  \sqrt{p_1p_3} e^{ j \theta_3 } \\
\sqrt{p_1p_2} e^{- j \theta_2 }  \hspace*{-1mm}&\hspace*{-1mm}  p_{2}
 \hspace*{-1mm}&\hspace*{-1mm}  \sqrt{p_2p_3} e^{ j \sphericalangle q_{2,3} } \\

\sqrt{p_1p_3} e^{- j \theta_3 }  \hspace*{-1mm}&\hspace*{-1mm}  \sqrt{p_2p_3} e^{- j \sphericalangle q_{2,3} }
 \hspace*{-1mm}&\hspace*{-1mm}  p_{3 }
\end{bmatrix}
\label{N3}
\end{equation}
 as a function of unknowns $ \theta_2, \sphericalangle q_{2,3}$ and $ \theta_3$.
Now, using $BQ=0$ in  \eqref{BQ} where
$B=
\begin{bmatrix}
b_{1,1} \hspace*{-1mm}&\hspace*{-1mm}  - \alpha_{1,2}
 \hspace*{-1mm}&\hspace*{-1mm}  - \alpha_{1 ,3} \\
- \alpha_{1,2}^*  \hspace*{-1mm}&\hspace*{-1mm}  b _{2,2}
 \hspace*{-1mm}&\hspace*{-1mm}  - \alpha_{2 ,3} \\
-\alpha_{1 ,3}^*  \hspace*{-1mm}&\hspace*{-1mm}  - \alpha_{2 ,3}^*
 \hspace*{-1mm}&\hspace*{-1mm}  b_{3 ,3}
\end{bmatrix},
$
for $[BQ]_{1,1}=0$ and $[BQ]_{1,2}=0$, we can  respectively write
\begin{eqnarray}
b_{1,1} p_1 - \alpha_{1,2} \sqrt{p_1,p_2} e^{- j \theta_2 } - \alpha_{1,3} \sqrt{p_1,p_3} e^{- j \theta_3 } =0,
\\
b_{1,1} \sqrt{p_1,p_2} e^{- j \theta_2 } - \alpha_{1,2} p_2 - \alpha_{1,3} \sqrt{p_2,p_3} e^{- j \sphericalangle q_{2,3} } =0.
\end{eqnarray}
Note that other entries of $BQ=0$  can not yield independent equations.
The above equations yield
\begin{align}
b_{1,1} &= \alpha_{1,2} \sqrt{\frac{p_2}{p_1}} e^{- j \theta_2 } + \alpha_{1,3} \sqrt{\frac{p_3}{p_1}} e^{- j \theta_3 },
\label{b11}
\\
b_{1,1} & = \alpha_{1,2} \sqrt{\frac{p_2}{p_1}} e^{- j \theta_2 } + \alpha_{1,3} \sqrt{\frac{p_3}{p_1}} e^{- j \sphericalangle q_{2,3} - \theta_2}.
\label{b12}
\end{align}
Comparing (\ref{b11}) and (\ref{b12}), we  conclude that $\sphericalangle q_{2,3} = \theta_3 - \theta_2$. Thus, it remains two unknown primal variables i.e., $\theta_2, \theta_3$ and three dual variables i.e., $ b_{1,1}, b_{2,2}, b_{3,3}$. Using $\sphericalangle q_{2,3} = \theta_3 - \theta_2$, similar to $b_{1,1}$, we  can obtain $b_{2,2}$ and $b_{3,3}$ by considering other entries of $BQ=0$ as
\begin{eqnarray}
 &\hspace*{-5mm} b_{1,1} = \frac{| \alpha_{1,2} | \sqrt{ {p_2}} e^{ j (\sphericalangle \alpha_{1,2} - \theta_2) } + | \alpha_{1,3} | \sqrt{ {p_3}} e^{ j (\sphericalangle \alpha_{1,3} - \theta_3)}}{\sqrt{p_1}},
\label{b112}\\
 &\hspace*{-5mm} b_{2,2} = \frac{| \alpha_{1,2} | \sqrt{p_1}  e^{ j (\theta_2 - \sphericalangle \alpha_{1,2}) } + | \alpha_{2,3} | \sqrt{ p_3 } e^{ j (\sphericalangle \alpha_{2,3} + \theta_2 - \theta_3)}}{\sqrt{ p_2 }},
\label{b22}
\\
 &\hspace*{-5mm} b_{3,3} = \frac{| \alpha_{1,3} | \sqrt{ {p_1} } e^{ j (\theta_3 - \sphericalangle \alpha_{1,3}) } + | \alpha_{2,3} | \sqrt{ {p_2}} e^{ j (- \sphericalangle \alpha_{2,3} + \theta_3 - \theta_2)}}{\sqrt {p_3}}.
\label{b33}
\end{eqnarray}
The dual variables in \eqref{b112}, \eqref{b22} and \eqref{b33} are expressed in terms of  $\theta_2$ and $\theta_3$. Thus the problem (P1) is converted to finding two primal variables for $N=3$. Since $b_{i,i}$ is real, the imaginary parts of \eqref{b112}-\eqref{b22} are zero, i.e.,
\begin{align}
\!\!\!\!\! {\sqrt{p_2}}| \alpha_{1,2}| \sin(\theta_2 -\sphericalangle\alpha_{1,2}) &= {\sqrt{p_3}} | \alpha_{1,3}| \sin(\sphericalangle \alpha_{1,3} - \theta_3),
\label{pppp1}
\\
\!\!\!\!\! {\sqrt{p_1}}| \alpha_{1,2}| \sin(\sphericalangle \alpha_{1,2}-\theta_2) &= {\sqrt{p_3}}| \alpha_{2,3}| \sin(\sphericalangle \alpha_{2,3} + \theta_2 - \theta_3).\!\!\!\!\!
\label{pppp2}
\end{align}
The imaginary part of \eqref{b33} leads to another equation which is not  linearly independent with \eqref{pppp1} and \eqref{pppp2}. The non-linear equations (\ref{pppp1}) and (\ref{pppp2}) can be solved using numerical methods to find the optimal values of $\theta_2$ and $\theta_3$.

\section{ Suboptimal Solutions for General Case}\label{sec:sub}
According to \eqref{BQ}, $Q$ is in the null space of the hermitian positive semi-definite matrix $B$ with off-diagonal entries, $b_{i,j} = - \alpha_{i,j}, i \neq j$. However, the diagonal values, $b_{i,i}, i=1, 2, \cdots, N$, are still unknown. To find a suboptimal solution and motivated by \eqref{ppp1}, we  approximate $b_{i,i}$ by   $b_{i,i}\simeq \sum_{j=1 , i \neq j}^{N} \frac{\sqrt{p_j}}{\sqrt{p_i}} | \alpha_{i,j} |$; the cost function is maximized if we could set $ e^{ j (\sphericalangle \alpha_{i,j} - \sphericalangle q_{i,j})}=1$. Since $B$ is approximated only by increasing its diagonal entries, it is obvious that the resulting approximated matrix $B$ becomes full rank and remains positive semi-definite; thus its null space becomes empty.
However, we expect that the eigenvector corresponding to the smallest eigenvalue of the approximated $B$ gives an accurate approximation for the basis of the null-space of $B$. Therefore using Lemma~\ref{lemma: 2},  since the rank of $Q$ is one, we  use the eigenvector associated to the smallest eigenvalue of the approximated $B$ denoted by $[w_{1,1},w_{1,2},\cdots,w_{1,N}]^T$  and define a closed-form suboptimal solution for (P1) as
\begin{equation}
\mathbf{w}_{\textrm{sub}1} =[
\sqrt{p_1 }\frac{\omega_{1,1}}{| \omega_{1,1} |}
,
\sqrt{p_2}\frac{\omega_{1,2}}{| \omega_{1,2} |}
,
\cdots
,
\sqrt{p_N}\frac{\omega_{1,N}}{| \omega_{1,N} |}]^T.
\label{sub1}
\end{equation}

An alternative suboptimal solution can be proposed by approximating $B = C - H^H H$ with another positive semi-definite matrix. Indeed, since $C$ is diagonal, $c_i$s can be approximated with dominant eigenvalue of $H^HH$ i.e., $c_i=\xi_1, i=1, ..., N$. In this case, the resulting $B = C - H^H H$ is a positive semi-definite matrix with at least one zero eigenvalue which its corresponding eigenvector is equivalent to the eigenvector corresponding to the dominant eigenvalue of $H^HH$ due to the diagonal structure of $C$. Thus, Lemma~\ref{lemma: 2}  yields another heuristic suboptimal solution
\begin{equation}
\mathbf{w}_{\textrm{sub}2} =
[
\sqrt{p_1 }\frac{v_{N,1}}{| v_{N,1} |}
,
\sqrt{p_2}\frac{v_{N,2}}{| v_{N,2} |}
,\cdots
\\
\sqrt{p_N}\frac{v_{N,N}}{| v_{N,N} |}
]^T
\label{sub2}
\end{equation}
where $[v_{N,1},v_{N,2},\cdots,v_{N,N}]^T$ is the eigenvector corresponding to the largest eigenvalue of $H^HH$. Compared with the optimal solution under SPC with the beam weight vector as $\sqrt{p_t} \mathbf{v_{\mathbf{N}}}$, we can see that the only difference between optimal beam weight of $i^{th}$ antenna with SPC i.e., $w_{1,i} =\sqrt{p_t}v_{N,i}$ and proposed sub-optimal solution in (\ref{sub2}) i.e., $w_{1,i} =\sqrt{p_i} \frac{v_{N,i}}{| v_{N,i} |}$ is in  the power allocations.

Lemma~\ref{lemma: 2}, also reveals that the optimal solution requires a single mode beamforming which allows a common data to be broadcasted to all receivers. Note that, it has been established that such a broadcasting, maximizes the sum rate under PAC.

\section{Simulations}\label{sec:sim}
In this section, we provide simulation results to validate our analytical results in previous sections. We consider a circular cell with $\mathcal{R}=15m$ overlaid by $K=10$ uniformly distributed ERs. The channel from the ET to
the $k^{th}$ ER is modelled as $\mathbf{h_k }=\sqrt{0.01 d_k^{-v}} \mathbf{\tilde{h}_k}$, where $d_k \in [1,\mathcal{R}]$ is the distance between ET and $k^{th}$ ER, $v=3$ is the path-loss exponent, and $ \mathbf{\tilde{h}_k}$ is the vector of small-scale Rayleigh fading coefficients with complex gaussian distribution with zero mean and unit variance as in \cite{deb}. We evaluate the average performance over 1000 ER locations and channel realizations. The total transmit power at the ET is set to $p_t=1 watt$ and we also set the transmit power limitation of each  antenna as $p_i = \frac{p_t}{N}$.

Here, we assume the Sum-energy maximization problem under different power constraints for $K=10$ ERs. Figure \ref{figg1} illustrates the results of the sum of harvested power versus the number of transmitted antennas. In the PAC (Beamforming) case, we applied CVX to solve the problem (P1), but in the PAC (independent) case, each transmit antenna has its own power budget and acts independently and the transmit strategy is isotropic. This constraint is equivalent to having a diagonal input covariance; i. e., $Q = \textrm{diag} \lbrace p_1, p_2, \cdots, p_{N}\rbrace$.
Comparing to the independent case, it can be observed that employing beamforming to create correlation among the transmit signals, significantly increases the sum of harvested power in both cases of sum power and PAC. For example, for $N=5$, the sum power for independent case is $1.22$mw, while it is $4.59$mw and $ 5.39$mw for PAC and SPC cases, respectively. Note that single-mode beamforming introduces complete correlation among the signals from different antennas since all antennas send the same symbol, with different weights \cite{miso}. However, it is obvious that under the SPC without PAC, power allocation can further increases the sum of harvested power, which can be seen in figure \ref{figg1}.
\begin{figure}
\centering
\includegraphics[width=.9\columnwidth]{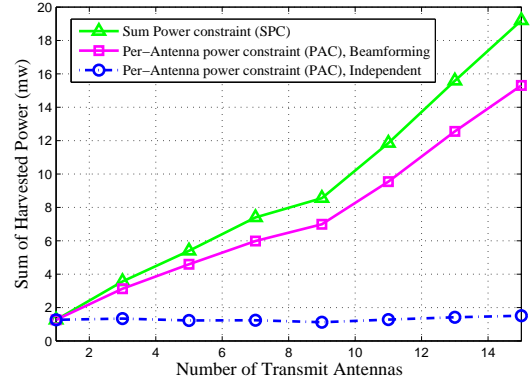}
\caption{ Sum of harvested power versus the number of transmit antennas for $K = 10$ under different power constraints}
\centering
\label{figg1}
\centering
\end{figure}

Fig. \ref{figg2} illustrates the results of problem (P1) which are obtained by 1) Optimal solution using CVX; 2) Proposed sub-optimal 1 solution in (\ref{sub1}), 3) Proposed sub-optimal 2 solution in (\ref{sub2}) and 4) Without beamforming. We see that our proposed solutions are matched to the optimum results from CVX and the difference is so small and negligible. These results show that although only the phase of $Q$ in the proposed sub-optimal 2 equals to the phase of optimal $Q$ in SPC problem, the numerical sub-optimum results are matched with high precision to the one that has been achieved by CVX. This indicates that the power constraint has a minor role in determining the phase of the optimal solution.
Magnifying the plots in Fig. (\ref{figg2}) shows that the performance of the proposed sub-optimal 1 solution is better than sub-optimal 2 method, which was predictable. Since, in the sub-optimal 1 method, we have used an estimated version of (\ref{ppp1}), while it has been ignored in sub-optimal 2 solution. In order to show the performance of beamforming strategy obtained by our sub-optimal solutions, the scheme without beamforming strategy is also considered. In this strategy, each transmit antenna sends the signal with random phase. Fig. (\ref{figg2}) shows the significant performance gains by beamforming with proposed sub-optimal methods in comparison to the scheme without beamforming. For example, for $N=5$, the sum power for the scheme without beamforming is $1.2834 $mw while it is $4.4976$mw and $4.4936$mw for proposed sub-optimal 1 and proposed sub-optimal 2 solutions, respectively.

The computational complexity of the methods significantly affects their potential applicability to the real-time or even off-line systems \cite{omid}. Thus, a comparison between the computational time for sub-optimal 1 solution and CVX numerical solution is performed and demonstrated in table \ref{table2} as another consideration about our proposed sub-optimal solutions.  Table \ref{table2} shows that the computational time for numerical solution is much larger than that of sub-optimal solutions. For example, for $N=15$, the computational time of CVX numerical solution is around 273 msec while it is 0.46 msec and 0.26 msec for proposed sub-optimal 1 and proposed sub-optimal 2 respectively on a desktop computer with 2.5 Ghz CPU and 4GB RAM. 

\begin{figure}
\centering
\includegraphics[width=.9\columnwidth]{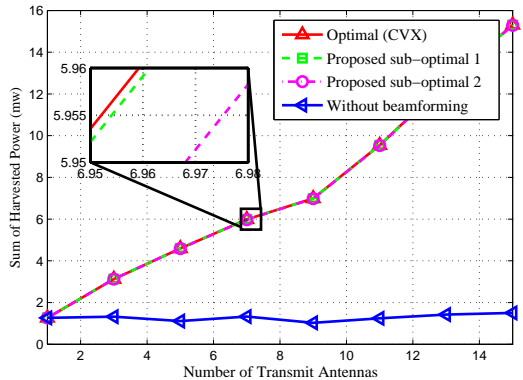}
\caption{ Sum of harvested power versus the number of transmit antennas for $K = 10$. }
\centering
\label{figg2}
\centering
\end{figure}

\begin{table}
\centering
\caption{Computational Time (milli sec) }\label{table2}
\begin{tabular}{c|c|c|c|c|c|c}
$N$& $5$&$10$&$15$&$20$&$25$
\\
\hline
Optimal numerical solution (CVX) &  244 & 246 &  273& 331 & 372
\\
\hline
Proposed sub-optimal 1 &0.30 &  0.35 &  0.46 & 0.76  &  1.02
\\
\hline
Proposed sub-optimal 2 &0.13 &  0.18 &  0.26 & 0.39  &  0.52
\end{tabular}
\end{table}

\section{Discussion and Conclusion }\label{sec:con}
In this paper, we have studied the EB problem with PAC in a WPT system; where a multi antenna ET transfers wireless energy to ERs which are randomly placed within the cell area. In the case of sum energy maximization, we have shown that the optimal transmit covariance matrix is rank-one and only the phases of the beamforming vector weights depend on the channel coefficients; however, their amplitudes are independent of the channel and depend only on the PACs.

Problem (P1) aims to find a semi-definite matrix $Q$ containing $N \times N$ complex entries. We reduced this problem to $N-1$  equations in  \eqref{ppppN} which can be solved numerically for  $N-1$ real unknowns $\{\theta_i\}_{i=2}^N$. In addition, we have proposed two accurate sub-optimal solutions in \eqref{sub1} and \eqref{sub2}. Simulation results show that these sub-optimal solutions are matched closely to optimal values  obtaining by optimization tools (See Fig. \ref{figg2}). Either \eqref{sub1} or \eqref{sub2} can be used for initialization of numerical algorithms to solve \eqref{ppppN}. To show the efficiency of our proposed methods, we compare the CPU run times of different algorithms which shows the superiority of the proposed solutions in the case of computational costs (See table \ref{table2}).

As mentioned earlier, a practical scenario for (P1) is in ET where the transmitted antennas are at different physical locations. Lemma~\ref{lemma: 2} implies that only one single energy beam is used for collaborative EB at different antennas. Therefore, they only need to store one common pseudo-random energy signal and the network coordinator only needs to send the optimum phase $\theta_i$, to the $i$th antenna for $ i=2,\cdots,N$. In this case a network coordinator is responsible for 1) collecting the information from all the distributed transmitters 2) finding the optimal phases and 2) sending the optimum $\theta_i$ to individual distributed transmitters.


\appendices

\bibliographystyle{IEEETran}
\bibliography{Reference}

\end{document}